\documentclass{kybernetika}
\pdfoutput=1 %

\usepackage[sort,noadjust]{cite}

\usepackage{microtype}
\usepackage{amsmath, amssymb, amsthm, xfrac}
\usepackage{graphicx, url, xcolor}

\usepackage{centernot}

\usepackage{fancyvrb}

\usepackage{wrapfig}

\usepackage{hyperref}
\usepackage[sort]{cleveref}

\usepackage[inline]{enumitem}

\usepackage{dsfont}
\usepackage[mathscr]{eucal}
\newcommand{\BB}[1]{\mathds{#1}}
\newcommand{\SR}[1]{\mathscr{#1}}

\usepackage{sansmath}

\newcommand{\PL}{\text{\textsl{\L}}}

\usepackage{tikz}
\newcommand{\bul}[1]{%
\scalebox{0.7}{\tikz[baseline={([yshift={2.5pt}] current bounding box.south)}]%
  \node [outer sep=0pt, inner sep=2pt, fill=white, draw, circle] {\bfseries#1};%
}}

\theoremstyle{definition}
\newtheorem{theorem}{Theorem}[section]
\newtheorem{lemma}[theorem]{Lemma}

\newtheorem{corollary}[theorem]{Corollary}
\newtheorem{remark}[theorem]{Remark}

\newtheorem{example}[theorem]{Example}
\newtheorem{definition}[theorem]{Definition}

\newtheorem{problem}[theorem]{Problem}
\newtheorem{question}[theorem]{Question}
\newtheorem{conjecture}[theorem]{Conjecture}

\usepackage{mathtools}
\newcommand{\defas}{\coloneqq}

\newcommand{\CIperp}{\mathrel{\text{$\perp\mkern-10mu\perp$}}}
\usepackage{ifthen}
\usepackage{listofitems}
\newcommand{\CI}[2][\CIperp]{{%
  \setsepchar{{=}/{|}/{:}}
  \ignoreemptyitems
  \readlist*\mylist{#2}
  \ifthenelse{\listlen\mylist[] = 2}{\mylist[2]}{}%
  [\mylist[1,1,1] #1 \ifthenelse{\listlen\mylist[1,1] = 2}{\mylist[1,1,2]}{\mylist[1,1,1]}%
  \ifthenelse{\listlen\mylist[1] = 2}{{} \mid \mylist[1,2]}{}]%
}}

\usepackage{xstring}

\newcommand{\Set}[1]{\{\,#1\,\}}
\newcommand{\SL}[1]{\text{\sffamily\slshape #1\/}}

\newcommand{\ol}[1]{\overline{#1}}
\newcommand{\eps}{\varepsilon}

\begin{document}
\pagestyle{myheadings}

\title{On the Intersection and Composition \\
properties of conditional independence}

\author{Tobias Boege}

\contact{Tobias}{Boege}%
{Department of Mathematics and Statistics,
UiT -- The Arctic University of Norway,
N-9037 Tromsø}%
{post@taboege.de}

\markboth{T.~Boege}{On Intersection and Composition}

\maketitle

\begin{abstract}
Compositional graphoids are fundamental discrete structures which appear
in probabilistic reasoning, particularly in the area of graphical models.
They are semigraphoids which satisfy the Intersection and Composition properties.
These important properties, however, are not enjoyed by general probability distributions.
This paper surveys what is known about them, providing systematic constructions
of examples and counterexamples as well as necessary and sufficient conditions.
Novel sufficient conditions for both properties are derived in the context of
discrete random variables via information-theoretic tools.
\end{abstract}

\keywords{semigraphoid, compositional graphoid, gaussoid, conditional
independence, entropy region, information inequality}

\classification{94A15 (primary); 62R01, 94A17 (secondary)}

\section{Introduction}
\label{sec:Intro}

One of the most fundamental aspects one could aim to understand about
a complex system is its \emph{dependence structure}: Which observables
depend on others? How many degrees of freedom does the vector of
observations have as the system evolves?
Insights about the dependence structure are not strictly required to
tackle more advanced questions about numerical simulation, branching
behavior or long-term prediction of the system but they contribute to
a theoretical understanding and efficient implementation.
The study of dependence (or \emph{special position}) is attractive
also because the concept is ubiquitous in mathematics and computer
science, from dimensions in linear algebra and algebraic geometry
\cite{CoxLittleOShea,Oxley} over dependence in statistics and graphs
\cite{Lauritzen,Studeny} to functional dependence in cryptography
and database theory \cite{InfoCrypto,MaierDB}.
Complex interactions arise in each of these settings and it is far
from trivial to understand the resulting dynamics of dependence,
as evidenced by universality \cite{BokowskiSturmfels} or undecidability
results \cite{KuehneYashfe,Li,Yashfe}.
The investigation of dependence as an abstract relation in various
settings is rewarded by inner-mathematical connections, for example
Matúš's observation \cite[Theorem~2]{MatusAscending} that every universally
valid implication among stochastic independence statements about random
variables also holds for linear independence of vectors (in any vector
space over a field). The dependence theory of random variables turns out
to occupy a sweet spot between usefulness and complication, and it
generalizes reasoning about special position in vector spaces.

The present paper deals with two implications, the \emph{Intersection}
and \emph{Composition} properties, which are important in statistical
modeling but are not universally valid. Their significance lies, for
one, in the applications that they enable. They could be compared to
Desargues's theorem in projective geometry. Desargues does not hold in
every projective geometry but if it does then the geometry is
coordinatized by a skew field~\cite[Chapter~II]{Artin} which enables
the use of algebraic methods.
On the other hand, Intersection and Composition are significant also
due to their closeness to the very basic, universal properties of
independence relations.
The study of such properties in statistics goes back to Dawid~\cite{DawidTheory,DawidOperations}
who advocated the use of \emph{conditional independence~(CI)} as a
foundational concept in statistical inference and described basic
relations among the valid CI~statements for any finite system~$N$
of jointly distributed random variables which became known later
as the \emph{semigraphoid properties}. They consist of the following
assertions and implications, for any four disjoint subsets
$I, J, K, L \subseteq N$; see~\cite[Section~2.2.2]{Studeny}:
\begin{description}[itemsep=0.1em, labelindent=\parindent]
\item[Triviality] $\CI{I:\emptyset|L}$,
\item[Symmetry] $\CI{I:J|L} \iff \CI{J:I|L}$,
\item[Decomposition] $\CI{I:JK|L} \implies \CI{I:J|L}$,
\item[Weak union] $\CI{I:JK|L} \implies \CI{I:K|JL}$,
\item[Contraction] $\CI{I:J|L} \land \CI{I:K|JL} \implies \CI{I:JK|L}$.
\end{description}

The Triviality axiom is inconsequential as it does not interact with the
other axioms in a way that produces other, non-trivial statements.
Throughout this paper, we accept the Symmetry axiom and formally identify
any CI~symbol $\CI{I:J|K}$ with its symmetric version $\CI{J:I|K}$. This
leaves Decomposition, Weak union and Contraction as the defining traits
of a semigraphoid. They can be restated more succinctly as an equivalence:
\begin{equation}
  \label{eq:Semgr}
  \CI{I:JK|L} \iff \CI{I:J|L} \land \CI{I:K|JL}.
\end{equation}
This property can be found in many notions of independence all across
mathematics and computer science. It holds for stochastic independence
as Dawid showed; for linear and algebraic independence and their vast
generalization of forking independence in geometric stability theory~\cite{Palacin}
(where this property is called ``Transitivity'');
and for various notions of separation in graphs \cite{UnifyingMarkov}
(and their appropriate generalizations to topological spaces).

Since the roles of $J$ and $K$ are interchangeable in the left-hand side
of \eqref{eq:Semgr}, we may consider a symmetrized version which is the
starting point for our investigation:
\[
  \CI{I:JK|L} \iff \begin{cases}
    \bul1\, \CI{I:J|L} \land \bul2\, \CI{I:K|JL} \land {} \\
    \bul3\, \CI{I:K|L} \land \bul4\, \CI{I:J|KL}.
  \end{cases}
\]
Under which circumstances are subsets of the statements on the right-hand
side sufficient to imply $\CI{I:JK|L}$ on the left-hand side, provided that
all random variables are discrete?
By Contraction, $\bul1 \land \bul2$ as well as $\bul3 \land \bul4$ are
always sufficient; hence, any 3-subset of $\Set{\bul1, \bul2, \bul3,\bul4}$
is sufficient. Up to interchanging $J$ and $K$, this leaves only three
configurations of the 2-element subsets to consider:
\begin{itemize}[itemsep=0.1em, labelindent=\parindent]
\item The implication $\bul1 \land \bul3 \implies \CI{I:JK|L}$ is the
  converse of (the symmetrized version of) Decomposition, called
  \emph{Composition}.
\item Similarly, $\bul2 \land \bul4 \implies \CI{I:JK|L}$ is the converse
  of (the symmetrized version of) Weak union and is called \emph{Intersection}.
\item Finally, the two symmetric implications $\bul1 \land \bul4 \implies
  \CI{I:JK|L}$ and $\bul2 \land \bul3 \implies \CI{I:JK|L}$ seem to be
  almost entirely disregarded in the literature, to the point where we
  could not find an established name for these implications.
\end{itemize}

The focus of this paper is on sufficient conditions for Intersection
and Composition; the nameless third implication is only briefly discussed
in \Cref{sec:Remarks}.
Unlike the semigraphoid properties, Intersection and Composition are not
universally valid: there exist discrete probability distributions which
satisfy the premises but not the conclusion~$\CI{I:JK|L}$.
Nevertheless, they can be verified for several families of \emph{graphical
models} (see \cite{UnifyingMarkov}) which play a prominent role in
applications.
Intersection classically appears as a technical condition which ensures
the equivalence of different Markov properties of graphical models (see
\cite[Theorem~3.7]{Lauritzen}). It~also guarantees the uniqueness of
Markov boundaries by \cite{PearlPaz} and drives certain identifiability
results described in \cite{Peters}. The Composition property is needed
in the correctness proof of the IAMB algorithm to find Markov boundaries;
cf.~\cite{MarkovBoundary}.
Continuing this line of work, more recent research of Amini, Aragam, and
Zhou~\cite{Bryon} seeks to decouple structure learning algorithms from the
graphical representation and faithfulness assumptions to generalize them
to situations in which only formal properties of the independence model,
such as Intersection and Composition, are assumed. This~has renewed
interest in sufficient conditions under which these properties hold.

The remainder of this paper is organized as follows. \Cref{sec:Reduction}
performs routine manipulations to reduce Intersection and Composition to
a standard form in which they turn out to be logical converses. %
\Cref{sec:Inter,sec:Compo} survey known sufficient conditions for
Intersection and Composition, respectively, discuss some interesting
example classes, and derive a set of new sufficient conditions.
Further remarks are collected in \Cref{sec:Remarks}.

\subsection*{Notational conventions}

Our notation for conditional independence statements largely follows the
standard reference \cite{Studeny}. In particular, $N$ is a finite set
indexing a system of jointly distributed random variables. Subsets of $N$
are usually called $I, J, K, L, \ldots$ and elements $i, j, k, l, \ldots$.
An element $i \in N$ also denotes the singleton subset $\Set{i} \subseteq N$.
Union of subsets of $N$ is abbreviated to $IJ = I \cup J$. A CI~statement
$\CI{I:J|K}$ is read as ``$I$ is independent of $J$ given~$K$''.
In \Cref{sec:Inter,sec:Compo} we work concretely with four discrete random
variables denoted $\SL A, \SL X, \SL Y, \SL G$.
Throughout we employ concepts such as entropy and conditional mutual
information from Shannon theory for which \cite{Yeung} is an~accessible~reference.
The~use of \emph{information diagrams} like the one in \Cref{fig:InfoDiagram}
often elucidates computations with Shannon entropies; see
\cite[Section~6.5]{Yeung} for an explanation of this method.

\begin{figure}
\begin{center}
\scalebox{0.7}{%
\begin{tikzpicture}[scale=0.6]
\node (a) at (0,0) {};
\node (b) at (4,0) {};
\node (c) at (2,3.464) {};
\draw (a) circle (3.7);
\draw (b) circle (3.7);
\draw (c) circle (3.7);

\node[shift={(0,0.9)}] (A|XY) at (c) {\scriptsize$H(\SL A\mid \SL X,\SL Y)$};
\node[shift={(-0.9,-0.4)}] (X|AY) at (a) {\scriptsize$H(\SL X\mid \SL A,\SL Y)$};
\node[shift={(0.9,-0.4)}] (Y|AX)  at (b)  {\scriptsize$H(\SL Y\mid \SL A,\SL X)$};
\node (A:X:Y) at (2,1.1) {\scriptsize$I(\SL A:\SL X:\SL Y)$};
\node (X:Y|A) at (2,-1.2) {\scriptsize$I(\SL X:\SL Y\mid\SL A)$};
\node (A:X|Y) at (-0.1,2.6) {\scriptsize$I(\SL A:\SL X\mid\SL Y)$};
\node (A:Y|X) at (4.1,2.6) {\scriptsize$I(\SL A:\SL Y\mid\SL X)$};

\node (A) at (2,8) {\large$\SL A$};
\node (X) at (-4,-3) {\large$\SL X$};
\node (Y) at (8,-3) {\large$\SL Y$};
\end{tikzpicture}}
\end{center}
\caption{The generic information diagram of three jointly distributed random
variables $\SL A, \SL X, \SL Y$. All of Shannon's information measures can be
expressed as linear combinations of these seven quantities.}
\label{fig:InfoDiagram}
\end{figure}

\section{Preliminary reductions}
\label{sec:Reduction}

It is well-known that any CI~statement $\CI{I:J|K}$ with pairwise disjoint
sets $I, J, K \subseteq N$ is equivalent modulo the semigraphoid axioms to
a conjunction of \emph{elementary} CI~statements:
\begin{equation}
  \label{eq:ElemCI}
  \CI{I:J|K} \;\iff\; \bigwedge_{i \in I} \bigwedge_{j \in J}
    \bigwedge_{K \subseteq L \subseteq IJK \setminus ij} \CI{i:j|L}.
\end{equation}
The proof of this fact merely combines Decomposition and Weak union
(with Symmetry) in one direction and Contraction in the other.
Since the semigraphoid axioms hold for any system of discrete random
variables, we may reformulate Intersection and Composition in terms
of elementary~CI using \eqref{eq:ElemCI} and arrive at the following
equivalent formulations:
\begin{alignat*}{2}
  \text{\bfseries Intersection}\quad && \CI{i:j|kL} \land \CI{i:k|jL} &\implies \CI{i:j|L}  \land \CI{i:k|L}, \\
  \text{\bfseries Composition}\quad  && \CI{i:j|L}  \land \CI{i:k|L}  &\implies \CI{i:j|kL} \land \CI{i:k|jL}.
\end{alignat*}
This is the form in which these properties are often presented in the
literature on gaussoids, such as \cite{LnenickaMatus}. This also shows
that Intersection and Composition are logical converses of each other
modulo the semigraphoid properties.

The final reduction concerns the conditioning set $L$ which is common
to all statements in the above CI~implication formulas. The ``full''
Intersection and Composition properties demand the above CI~implications
to hold for each choice of distinct $i, j, k \in N$ and $L \subseteq
N \setminus ijk$. Each quadruple $(i,j,k,L)$ encodes an \emph{instance}
of the property. In a given instance $(i,j,k,L)$, we may marginalize the
distribution to $ijkL$ and condition on~$L$. Thus, we arrive at the
following problem formulation which is addressed in this paper.

\begin{problem} \label{prob:Main}
For jointly distributed discrete random variables $(\SL A, \SL X, \SL Y)$,
find~sufficient conditions such that
\begin{alignat*}{2}
  \text{\bfseries Intersection}\quad && \CI{\SL A:\SL X|\SL Y} \land \CI{\SL A:\SL Y|\SL X} &\implies \CI{\SL A:\SL X|}  \land \CI{\SL A:\SL Y|}, \text{ respectively,}\\
  \text{\bfseries Composition}\quad  && \CI{\SL A:\SL X|}     \land \CI{\SL A:\SL Y|}     &\implies \CI{\SL A:\SL X|\SL Y} \land \CI{\SL A:\SL Y|\SL X}.
\end{alignat*}
\end{problem}

If $T(\SL A,\SL X,\SL Y)$ is a sufficient condition for Intersection or,
respectively, Composition to hold in a trivariate discrete distribution,
then a sufficient condition for the full Intersection or Composition
property is obtained as a conjunction of $T(i,j,k \mid L=\omega)$ over
all quadruples~$(i,j,k,L)$ and all events $\omega$ of $L$.

\begin{figure}
\begin{center}
\scalebox{0.7}{%
\begin{tikzpicture}[scale=0.6]
\node (a) at (0,0) {};
\node (b) at (4,0) {};
\node (c) at (2,3.464) {};
\draw (a) circle (3.7);
\draw (b) circle (3.7);
\draw (c) circle (3.7);

\node[shift={(0,0.9)}] (A|XY) at (c) {\Large$*$};
\node[shift={(-0.9,-0.4)}] (X|AY) at (a) {\Large$*$};
\node[shift={(0.9,-0.4)}] (Y|AX)  at (b)  {\Large$*$};
\node (A:X:Y) at (2,1.2) {\Large$h$};
\node (X:Y|A) at (2,-1.4) {\Large$g$};
\node (A:X|Y) at (-0.1,2.6) {\Large$0$};
\node (A:Y|X) at (4.1,2.6) {\Large$0$};

\node (A) at (2,8) {\large$\SL A$};
\node (X) at (-4,-3) {\large$\SL X$};
\node (Y) at (8,-3) {\large$\SL Y$};
\node (Inter) at (2,-5) {\Large\strut Premises of Intersection};
\end{tikzpicture}}
\quad
\scalebox{0.7}{%
\begin{tikzpicture}[scale=0.6]
\node (a) at (0,0) {};
\node (b) at (4,0) {};
\node (c) at (2,3.464) {};
\draw (a) circle (3.7);
\draw (b) circle (3.7);
\draw (c) circle (3.7);

\node[shift={(0,0.9)}] (A|XY) at (c) {\Large$*$};
\node[shift={(-0.9,-0.4)}] (X|AY) at (a) {\Large$*$};
\node[shift={(0.9,-0.4)}] (Y|AX)  at (b)  {\Large$*$};
\node (A:X:Y) at (1.85,1.2) {\Large$-f$};
\node (X:Y|A) at (2,-1.4) {\Large$g$};
\node (A:X|Y) at (-0.1,2.6) {\Large$f$};
\node (A:Y|X) at (4.1,2.6) {\Large$f$};

\node (A) at (2,8) {\large$\SL A$};
\node (X) at (-4,-3) {\large$\SL X$};
\node (Y) at (8,-3) {\large$\SL Y$};
\node (Compo) at (2,-5) {\Large\strut Premises of Composition};
\end{tikzpicture}}
\end{center}
\caption{Information diagrams assuming the premises of Intersection
and Composition, respectively. Both diagrams feature two non-negative
parameters: on the left side $g = I(\SL X:\SL Y\mid \SL A)$ and
$h = I(\SL A:\SL X:\SL Y)$; on the right side $g$ and
$f = I(\SL A:\SL X\mid \SL Y) = I(\SL A:\SL Y\mid \SL X)
= -I(\SL A:\SL X:\SL Y)$.
}
\label{fig:Assump}
\end{figure}

\Cref{fig:Assump} shows the premises of Intersection and Composition in
information diagrams. The conclusion $I(\SL A:\SL X, \SL Y)$ is equivalent
to the vanishing of the non-negative quantities~$h$ (for Intersection),
respectively~$f$ (for Composition). Note that in the case of Composition,
Shannon inequalities mandate that $g \ge f$ but no such inequality is
implied under the premises of Intersection.
We seek general conditions which guarantee $h = 0$, respectively $f = 0$,
when the respective premises hold but are not too severe when the premises
do not hold (in~which case the property holds vacuously).

\section{The Intersection property}
\label{sec:Inter}

The problem of finding sufficient conditions for the Intersection property
has received considerable attention from a variety of research communities.
The most widely known and the simplest general condition on a distribution
which ensures the full Intersection property is that the probability density
be strictly positive. This is sufficient but not necessary and, depending
on the application, may be too restrictive.

\subsection{Examples and non-examples}

The keyword \emph{graphoid} is helpful in locating examples of Intersection
in the literature: it~means a semigraphoid which satisfies Intersection.
The canonical examples are various Markov properties of graphical models,
including Bayesian networks, Markov networks \cite{UnifyingMarkov} as well
as $\ast$- and $C^\ast$-separation \cite{MaxLinearCI,Maxoids}.

\begin{example} \label{ex:Inter:Gaussoids}
A \emph{gaussoid} is a semigraphoid satisfying Intersection, Composition
and a further property called \emph{Weak transitivity}. This notion was
introduced by Lněnička and Matúš~\cite{LnenickaMatus} to model the
conditional independence structure of regular Gaussian random variables.
Chen~\cite[Section~2.2]{Xiangying} records that gaussoids furthermore
appear as the vanishing almost-principal quasi-minors of a polarity in a
Desarguesian projective~geometry, and as abstract orthogonality relations
on atoms in modular lattices. The interested reader is referred to
\cite{Xiangying} for detailed definitions and proofs.
\end{example}

We now focus on examples of the failure of Intersection.

\begin{example}[Three binary random variables] \label{ex:Inter:Binary}
The joint distribution of three binary random variables is given by eight
non-negative real numbers $p_{000}, p_{001}, \ldots, p_{111}$ which are
indexed by triples over the set $\Set{0,1}$ and sum to one. The set of all
such distributions is known as the probability simplex $\Delta(2,2,2)$.
A generic choice of these values leads to a distribution which does not
satisfy any CI~statement and therefore satisfies Intersection vacuously.
To violate Intersection, at least its premises must be fulfilled.
The set of such distributions is the intersection of $\Delta(2,2,2)$ with
an algebraic variety $V$ and its structure can be examined using primary
decomposition in \texttt{Macaulay2} (\cite{M2}) as described in \cite{GMAlgebra}.
\begin{Verbatim}[fontsize=\footnotesize, frame=leftline]
needsPackage "GraphicalModels";

R = markovRing(2,2,2);
I = conditionalIndependenceIdeal(R, {{{1},{2},{3}}, {{1},{3},{2}}});
J = conditionalIndependenceIdeal(R, {{{1},{2,3},{}}});
decompose(I:J)
\end{Verbatim}
The above decomposition describes the two irreducible components of $V$ in
$\Delta(2,2,2)$ on which there are distributions which violate Intersection.
They are given by the conditions
\begin{gather}
  \label{eq:Inter:FD:1}
  p_{000} = p_{011} = p_{100} = p_{111} = 0, \text{ or} \\
  \label{eq:Inter:FD:2}
  p_{001} = p_{010} = p_{101} = p_{110} = 0. \hphantom{\text{ or}}
\end{gather}
As expected, violations of Intersection can only occur on the boundary of
$\Delta(2,2,2)$ where the probability mass function has zeros and not all
of the eight joint events are possible.
Choosing for instance the zero pattern \eqref{eq:Inter:FD:2} leaves four
non-negative parameters $p_{000}, p_{011}, p_{100}, p_{111}$ which must sum
to~one. Choosing generic values for these probabilities yields a 3-parameter
family of distributions which satisfy the premises but not the conclusion
of Intersection; a different 3-parameter family is obtained analogously from~\eqref{eq:Inter:FD:1}.
The stipulation of ``generic values'' here is necessary: for special choices
of the remaining four probabilities (such as all of them equal to $\sfrac14$,
as in \Cref{ex:Inter:NonGK}) it happens that Intersection does hold. The
primary decomposition guarantees that these special distributions are confined
to a set of dimension at most two. Hence almost all (with respect to the
Lebesgue measure) distributions in our 3-parameter family violate Intersection.
\end{example}

\begin{example}[Functional dependencies] \label{ex:Inter:FD}
The random variable $\SL A$ depends functionally on $\SL X$ if the conditional
entropy $H(\SL A \mid \SL X)$ vanishes. This is equivalent to the existence of
a deterministic function~$f$ such that $\Pr[\SL A = f(\SL X)] = 1$, i.e.,
the value of $\SL X$ determines the outcome of $\SL A$ almost surely.
In this case (and if $\SL A$ is non-constant overall), the joint distribution
cannot be strictly positive.
Functional dependencies occur frequently in the context of relational databases
and may present themselves in measurements of physical quantities because of
the laws of nature.
If $\SL A$ functionally depends on $\SL Y$ and, separately, also functionally
depends on $\SL X$, then the conditional independencies $\CI{\SL A:\SL X|\SL Y}$
and $\CI{\SL A:\SL Y|\SL X}$ hold. It also follows that $A$ is a function of
the tuple $(\SL X, \SL Y)$. It is then an exercise in the use of the information
diagram method (see \Cref{fig:InfoDiagram}) that $I(\SL A : \SL X : \SL Y) =
H(\SL A)$. Thus, if $\SL A$ is also non-constant (hence has positive Shannon
entropy $H(\SL A)$), then the conclusion of Intersection is not satisfied.
\end{example}

\begin{remark}
Note that the conditions \eqref{eq:Inter:FD:1} and \eqref{eq:Inter:FD:2} in
\Cref{ex:Inter:Binary} enforce in both cases that $\SL X$ is a function of
$\SL Y$ and vice versa. %
It is possible to violate Intersection without any functional dependencies
in the distribution, but this requires larger state spaces.
\end{remark}

\begin{example}[Co-simple matroids] \label{ex:Inter:Matroids}
Let $(N, r)$ be a co-simple matroid; cf.~\cite{Oxley}. Then $r(N) =
r(N \setminus i) = r(N \setminus ij)$ for any $i, j \in N$. This implies
that $\CI{i:j|N \setminus ij}$ for all distinct $i, j \in N$. If the
full Intersection property holds for $r$, then for any $k \in N \setminus ij$
we may use $\CI{i:j|N\setminus ij} \land \CI{i:k|N\setminus ik}
\implies \CI{i:j|N\setminus ijk}$. Used inductively, this argument yields
that $r$ is totally independent, i.e., it satisfies $\CI{i:j|K}$ for all
distinct $i, j \in N$ and all $K \subseteq N \setminus ij$. Since every
$i \in N$ is simultaneously independent of and functionally dependent
on $N \setminus i$, we conclude that $r$ is the zero matroid.
\end{example}

\subsection{The Gács--Körner criterion}

The positivity of the entire distribution guarantees Intersection but is
unnecessarily restrictive. A more refined support condition has been
developed independently by groups of statisticians, information theorists
and algebraists. It is based on the following~concept.

\begin{definition}
Let $\SL X$ and $\SL Y$ be jointly distributed discrete random variables
with state spaces $Q_{\SL X}$ and $Q_{\SL Y}$, respectively. Their \emph{%
characteristic bipartite graph} $G(\SL X, \SL Y)$ is the bipartite graph
on $Q_{\SL X} \sqcup Q_{\SL Y}$ with an edge between events $x$ and $y$
if and only if $\Pr[\SL X=x, \SL Y=y] > 0$.
\end{definition}

\begin{figure}
\begin{center}
\scalebox{0.7}{%
\begin{tikzpicture}[scale=0.6]
\node (a) at (0,0) {};
\node (b) at (4,0) {};
\node (c) at (2,3.464) {};
\draw (a) circle (3.7);
\draw (b) circle (3.7);
\draw (c) circle (3.7);

\node[shift={(0,0.9)}] (A|XY) at (c) {\Large$0$};
\node[shift={(-0.9,-0.4)}] (X|AY) at (a) {\Large$*$};
\node[shift={(0.9,-0.4)}] (Y|AX)  at (b)  {\Large$*$};
\node (A:X:Y) at (2,1.1) {\Large$h$};
\node (X:Y|A) at (2,-1.4) {\Large$g$};
\node (A:X|Y) at (-0.1,2.6) {\Large$0$};
\node (A:Y|X) at (4.1,2.6) {\Large$0$};

\node (A) at (2,8) {\large$\SL G$};
\node (X) at (-4,-3) {\large$\SL X$};
\node (Y) at (8,-3) {\large$\SL Y$};
\end{tikzpicture}}
\end{center}
\caption{Information diagram of the Gács--Körner problem.}
\label{fig:GK}
\end{figure}

This graph appears in the work of Gács and Körner \cite{GacsKoerner} on
common information where it is used to construct a random variable
$\SL{GK}(\SL X, \SL Y)$ which solves the following optimization problem
aimed at extracting the maximum entropy of a random variable which is
simultaneously a function of $\SL X$ and of $\SL Y$:
\begin{equation}
  \label{eq:GK}
  \begin{aligned}
    \max \;& H(\SL G) \\
    \text{s.t.} \;& H(\SL G\mid \SL X) = H(\SL G\mid \SL Y) = 0.
  \end{aligned}
\end{equation}
As seen in \Cref{ex:Inter:FD} the functional dependence assumptions
imply $\CI{\SL G:\SL X|\SL Y}$ and $\CI{\SL G:\SL Y|\SL X}$. Recall the
decomposition of the mutual information $I(\SL X:\SL Y) =
{I(\SL X:\SL Y\mid \SL G)} + I(\SL G:\SL X:\SL Y)$. In \Cref{fig:GK} the
two summands on the right-hand side are denoted $g$ and $h$, respectively.
The diagram shows that $h = H(\SL G)$ is non-negative and~$g$, being a
conditional mutual information, is non-negative as well.
In the Gács--Körner problem $I(\SL X:\SL Y)$ is fixed and the objective
is to find a random variable $\SL G$ which maximizes the~$h$~part in this
decomposition.

The optimal value is known as the \emph{Gács--Körner common information}.
The solution $\SL{GK}(\SL X, \SL Y)$ has as its events the connected
components of $G(\SL X, \SL Y)$ and is specified as a function of $(\SL X, \SL Y)$
to evaluate to the connected component in which the outcomes of $\SL X$
and $\SL Y$ both lie. Since by construction $\Pr[\SL X=x, \SL Y=y] > 0$
if and only if $x$ and $y$ lie in the same connected component,
$G(\SL X, \SL Y)$ is well-defined and satisfies the functional dependence
constraints in \eqref{eq:GK}.
In our context, its significance lies in the following fact:

\begin{theorem} \label{thm:Inter:GK}
If $\CI{\SL A:\SL X|\SL Y}$ and $\CI{\SL A:\SL Y|\SL X}$, then
$\CI{\SL A:\SL X, \SL Y|\SL{GK}(\SL X,\SL Y)}$.
\end{theorem}

In information theory, this is sometimes called the \emph{double Markov
property} after Exercise 16.25 in the book of Csiszár and Körner~\cite{CsiszarKoerner}.

\begin{corollary}[Gács--Körner criterion] \label{cor:Inter:GK}
If $G(\SL X, \SL Y)$ is connected, then $\CI{\SL A:\SL X|\SL Y} \land
\CI{\SL A:\SL Y|\SL X} \implies \CI{\SL A:\SL X,\SL Y}$.
\end{corollary}

\begin{proof}
If the premises of Intersection hold, then \Cref{thm:Inter:GK} yields
$\CI{\SL A:\SL X,\SL Y|\SL{GK}(\SL X, \SL Y)}$. Since $G(\SL X, \SL Y)$
is connected, the random variable $\SL{GK}(\SL X, \SL Y)$ is a constant
and the conditional independence simplifies to $\CI{\SL A:\SL X, \SL Y}$,
the desired conclusion.
\end{proof}

\begin{remark}
The Gács--Körner condition depends only on the marginal distribution
of $(\SL X, \SL Y)$ and guarantees Intersection with respect to \emph{any}
discrete random variable~$\SL A$. By~requiring the maximum entropy
random variable $\SL G$ which is separately a function of $\SL X$ and of
$\SL Y$ to be constant, it is designed to forbid the class of counterexamples
to Intersection which was discussed in \Cref{ex:Inter:FD}. However, thanks
to \Cref{thm:Inter:GK} it is even sufficient for Intersection.
\end{remark}

This sufficient condition for one instance of Intersection indirectly
also targets the support of the distribution but instead of requiring
positivity everywhere, it only requires enough positivity on the marginal
distribution of $(\SL X, \SL Y)$ to make their characteristic bipartite
graph connected. Even when this property is required for all pairs of
random variables in a larger random vector, the resulting assumptions
are weaker than strictly positive support.
An equivalent condition in terms of $\sigma$-algebras is already present
in Dawid's foundational paper \cite{DawidOperations} and features in other
works under the name \emph{measurable separability}. San Martín, Mouchart
and Rolin \cite{Ignorable} provide an overview of the history of this idea
on the statistics side.

In algebraic statistics, a similar result is known as the Cartwright--Engström
conjecture which was recorded in \cite{AlgStat} and resolved by Fink in~\cite{Fink}.
Fink's theorem explains the combinatorial structure of the irreducible
decomposition of the binomial ideal corresponding to the premises of
Intersection. The irreducible components appearing in this decomposition
are associated to the \emph{admissible} bipartite graphs on $Q_{\SL X} \sqcup Q_{\SL Y}$,
i.e., bipartite graphs whose connected components are isomorphic to complete
bipartite graphs. Under this correspondence, a generic distribution of
$(\SL A, \SL X, \SL Y)$ in the component associated to an admissible
bipartite graph $G$ will have $G(\SL X, \SL Y) = G$. The only admissible
graph which is connected is the complete bipartite graph on $Q_{\SL X}
\sqcup Q_{\SL Y}$. A generic distribution in its component satisfies the
assumption of \Cref{cor:Inter:GK} and thus the Intersection property.
Generic distributions in all other components violate Intersection,
generalizing the computational results observed in \Cref{ex:Inter:Binary}.

\begin{example}[Incompleteness of the Gács--Körner criterion] \label{ex:Inter:NonGK}
The following table defines a joint distribution of four binary random
variables in which $\SL G$ is the Gács--Körner common information of
$\SL X$ and $\SL Y$. Since $\SL G$ is non-constant, the criterion of
\Cref{cor:Inter:GK} does not apply. Nevertheless, the distribution
satisfies $\CI{\SL A:\SL X, \SL Y}$ and therefore Intersection.
\begin{center}
\setlength{\tabcolsep}{1.5em}
\renewcommand{\arraystretch}{1.1}
\begin{tabular}{c|c|c|c||c}
$\SL A$ & $\SL X$ & $\SL Y$ & $\SL G$ & $\Pr$ \\ \hline
  $0$   &   $0$   &   $0$   &   $0$   & $\sfrac{1}{4}$ \\
  $0$   &   $1$   &   $1$   &   $1$   & $\sfrac{1}{4}$ \\
  $1$   &   $0$   &   $0$   &   $0$   & $\sfrac{1}{4}$ \\
  $1$   &   $1$   &   $1$   &   $1$   & $\sfrac{1}{4}$
\end{tabular}
\end{center}
\end{example}

\subsection{The conditional Ingleton criterion}
\label{sec:Inter:CondIng}

This section derives a set of new sufficient conditions for Intersection.
Like the Gács--Körner criterion above, they are formulated \emph{synthetically},
i.e., in terms of an auxiliary random variable~$\SL G$ which satisfies
additional CI~constraints with respect to $\SL A, \SL X, \SL Y$.
In this situation, the random variables are subject to powerful
information-theoretic inequalities. %
We~take advantage of recent work of Studený~\cite{StudenyIngleton}
which elucidates the connections between CI~implications on four
discrete random variables %
and special information-theoretic constraints known as
\emph{conditional Ingleton inequalities}.

We illustrate this approach with a concrete example. Let $\SL A, \SL X, \SL Y, \SL G$
be jointly distributed. The \emph{Ingleton expression} with respect to $\SL A$ and $\SL G$ is
\begin{gather}
  \label{eq:Ingleton}
  \begin{aligned}
  \square(\SL A, \SL G) &\defas H(\SL X, \SL G) + H(\SL Y, \SL G) + H(\SL X, \SL A) + H(\SL Y, \SL A) + H(\SL X, \SL Y) - {} \\
  &\hphantom{{}\defas{}} H(\SL A, \SL G) - H(\SL X) - H(\SL Y) - H(\SL X, \SL Y, \SL A) - H(\SL X, \SL Y, \SL G).
  \end{aligned}
\end{gather}
A celebrated result of Ingleton \cite{Ingleton} states that if $\SL A, \SL X, \SL Y, \SL G$
are subspaces of a finite-dimensional vector space and, instead of entropy,
the function $H$ in \eqref{eq:Ingleton} is interpreted as the dimension of
the span of its arguments, then $\square(\SL A, \SL G) \ge 0$. This inequality
does not hold in general for random variables but it becomes valid when certain
further assumptions are imposed. These assumptions give rise to what is called
a \emph{conditional Ingleton inequality}. For instance if the premises of
Intersection $\CI{\SL A:\SL X|\SL Y}$ and $\CI{\SL A:\SL Y|\SL X}$ hold then
\cite[Theorem~1 (3cI)]{StudenyIngleton} guarantees that $\square(\SL A, \SL G) \ge 0$.
One can verify the following equality by replacing each (conditional) mutual
information on the right-hand side by its expression in terms of joint entropies:
\[
  \square(\SL A, \SL G) = I(\SL X:\SL Y\mid \SL G) + I(\SL A:\SL X\mid \SL Y) + I(\SL A:\SL G) - I(\SL A:\SL X).
\]
Combining the CI~assumptions and the resulting non-negativity of $\square(\SL A, \SL G)$
yields
\[
  0 \le I(\SL A:\SL X) \le I(\SL X:\SL Y\mid \SL G) + I(\SL A:\SL G).
\]
Thus if additionally $\CI{\SL X:\SL Y|\SL G}$ and $\CI{\SL A:\SL G}$ hold,
then $\CI{\SL A:\SL X}$ is forced which implies $\CI{\SL A:\SL X, \SL Y}$
by the semigraphoid properties. This proves a CI~implication:
\[
  \CI{\SL A:\SL X|\SL Y} \land \CI{\SL A:\SL Y|\SL X} \land \CI{\SL X:\SL Y|\SL G} \land \CI{\SL A:\SL G} \implies \CI{\SL A:\SL X,\SL Y}
\]
which is valid for any four jointly distributed discrete random variables.
It involves the premises of Intersection and implies the conclusion.
This implication is recorded as rule (I:4) in \cite{StudenyIngleton} and,
following the same template, Studený derived many more. Our new sufficient
conditions for Intersection are based on these implications.

\begin{theorem}[Conditional Ingleton criterion] \label{thm:Inter:CondIng}
Let $\SL A, \SL X, \SL Y$ be jointly distributed discrete random
variables. Suppose that there exists a discrete $\SL G$ jointly distributed
with $\SL A, \SL X, \SL Y$ satisfying any of the following conditions:
\begin{enumerate}[label=(\roman*), itemsep=0.2em, leftmargin=4em]
\item\label{thm:Inter:CondIng:X}$\CI{\SL A:\SL G}$ and $\CI{\SL X:\SL Y|\SL G}$,
\item\label{thm:Inter:CondIng:Y}$\CI{\SL X:\SL G}$ and $\CI{\SL A:\SL Y|\SL G}$, or
\item\label{thm:Inter:CondIng:Z}$\CI{\SL Y:\SL G}$ and $\CI{\SL A:\SL X|\SL G}$.
\end{enumerate}
Then $\CI{\SL A:\SL X|\SL Y} \land \CI{\SL A: \SL Y|\SL X}
\implies \CI{\SL A:\SL X, \SL Y}$ holds.
\end{theorem}

\begin{proof}
Given $\CI{\SL A:\SL X|\SL Y}$ and $\CI{\SL A:\SL Y|\SL X}$, the
conditions~\ref{thm:Inter:CondIng:Y} and~\ref{thm:Inter:CondIng:Z}
are symmetric with respect to exchanging $\SL X$ and $\SL Y$ and
both follow from rule~(I:2) in \cite{StudenyIngleton}.
Condition~\ref{thm:Inter:CondIng:X} is covered by rule (I:4).
\end{proof}

In order to compare \Cref{thm:Inter:GK,thm:Inter:CondIng}, assume that $\SL G$
is a function of $\SL X$ and of $\SL Y$. In condition~\ref{thm:Inter:CondIng:Y},
the independence assumption $\CI{\SL X:\SL G}$ then implies that $\SL G$ is
constant and hence the further assumption $\CI{\SL A:\SL Y|\SL G}$ simplifies
to the desired conclusion already; a similar argument applies to condition~\ref{thm:Inter:CondIng:Z}.
Regarding condition~\ref{thm:Inter:CondIng:X}, the assumption $\CI{\SL X:\SL Y|\SL G}$
is equivalent to $H(\SL G) = I(\SL X:\SL Y)$. This is highly unusual when $\SL G$
is a function of $\SL X$ and of~$\SL Y$. Indeed, the Gács--Körner theorem
(see, e.g., \cite{Laszlo}) asserts that this can only happen if the probability
table of $(\SL X, \SL Y)$ can be brought into block-diagonal form by permutations
of its rows and columns and each block has rank~one.
It appears that the two criteria in \Cref{thm:Inter:GK,thm:Inter:CondIng} are
complementary and neither implies the other.
Note that the distribution given in \Cref{ex:Inter:NonGK} does not satisfy the
Gács--Körner criterion but does satisfy the conditional Ingleton criterion
\Cref{thm:Inter:CondIng}~\ref{thm:Inter:CondIng:X} since $\SL X = \SL Y = \SL G$
are functionally equivalent and the marginal $(\SL A, \SL G)$ is uniform.

\section{The Composition property}
\label{sec:Compo}

The previous section showed that the Intersection property is well-studied.
By comparison, not much is known about the failure modes of Composition.
We again begin by examining examples and show how some approaches that were
successful for Intersection cannot succeed for Composition.

\subsection{Examples and non-examples}

To start, Gaussian distributions satisfy Composition by~\cite[Corollary~2.4]{Studeny}.
This is true even when the covariance matrix is singular (in which case
the distribution need not satisfy Intersection). The gaussoids discussed
in \Cref{ex:Inter:Gaussoids} satisfy Composition as well. All types of
graphical models which can be faithfully represented by Gaussians thus
inherit the Composition property; some of them are also faithfully
representable by discrete random vectors.

\begin{example} \label{ex:Compo:MTP2}
In the discrete setting, a known sufficient condition is \emph{multivariate
total positivity of order 2} ($\mathrm{MTP}_2$) which is a type of
log-supermodularity condition on the density function. Fallat et~al.~\cite{MTP2}
show that $\mathrm{MTP}_2$ implies upward stability (i.e., $\CI{I:J|K} \implies
\CI{I:J|L}$ for any $L \supseteq K$), which is far stronger than Composition.
\end{example}

\begin{example} \label{ex:Compo:PartialOrtho}
\emph{Partial orthogonality} is a conditional independence-like relation
on vector configurations in a Hilbert space which satisfies the Composition
property. We refer to \cite{CompoML} for the precise definition and further
references. If the underlying set of vectors is linearly independent, also
Intersection holds. Partial orthogonality is used in machine learning as a
measure of semantic independence. In this setting, the ambient dimension
is usually much smaller than the number of vectors of interest, so
Intersection may not hold but Composition~does.
\end{example}

\begin{example}[Non-negative matrix factorization] \label{ex:Compo:MatrixFactor}
Any system of random variables $(\SL X, \SL Y)$ may be extended by $\SL A$
so that $\CI{\SL X:\SL Y|\SL A}$ holds. Following Reichenbach~\cite[Section~19]{Reichenbach}
such a variable $\SL A$ is called a \emph{common cause} (in particular if
$\SL X$ and $\SL Y$ are dependent). For this triple, a simple calculation
with the semigraphoid properties shows that all three instances of Composition
will hold.
A trivial choice of common cause is $\SL A = (\SL X, \SL Y)$ but there are
many possible choices: one for each \emph{non-negative matrix factorization}
of the joint probability table~$P$ of $\SL X$ and $\SL Y$. Indeed, write
\[
  P = \sum_{i=1}^k \gamma_i P^{(i)},
\]
where $P^{(i)}$ are non-negative matrices of rank one with $\lVert
P^{(i)}\rVert_1 = 1$, and $\gamma_i > 0$. Then the joint distribution
\[
  \Pr[\SL A = i, \SL X = x, \SL Y = y] = \gamma_i P^{(i)}_{xy}
\]
satisfies $\CI{\SL X:\SL Y|\SL A}$ and has the same $(\SL X, \SL Y)$-marginal
distribution as before. If the number of states $k$ is minimal, $\SL A$
corresponds to a non-negative rank factorization \cite{KaieFactor}. If~the
factorization is instead chosen for minimal entropy of the common cause,
we recover the \emph{common entropy} extension of $\SL X$ and $\SL Y$~\cite{ExactCommon}.
\end{example}

We now switch to examples of the failure of Composition.

\begin{example}[Three binary random variables] \label{ex:Compo:Binary}
The assumptions $\CI{\SL A:\SL X}$ and $\CI{\SL A:\SL Y}$ define a
\emph{marginal independence model} which can be easily parametrized
using the results of \cite{Kirkup}. For binary states, this
parametrization is as follows:
\begin{equation}
\label{eq:Compo:Binary}
\begin{aligned}
p_{000} &= \alpha \beta \gamma - \delta, &\qquad p_{100} &= \ol\alpha \beta \gamma - \eps, \\
p_{001} &= \alpha \beta \ol\gamma + \delta, &\qquad p_{101} &= \ol\alpha \beta \ol\gamma + \eps, \\
p_{010} &= \alpha \ol\beta \gamma + \delta, &\qquad p_{110} &= \ol\alpha \ol\beta \gamma + \eps, \\
p_{011} &= \alpha \ol\beta \ol\gamma - \delta, &\qquad p_{111} &= \ol\alpha \ol\beta \ol\gamma - \eps,
\end{aligned}
\end{equation}
where $\alpha, \beta, \gamma \in (0,1)$ and $\ol x = 1-x$;
the values of $\delta$ and $\eps$ are subject to the conditions that
all these probabilities must be non-negative. If $\alpha = \beta =
\gamma = \sfrac12$, $\delta = 0$ and $\eps > 0$ is small, then the
parametrization defines a probability distribution which satisfies
$\CI{\SL A:\SL X}$ and $\CI{\SL A:\SL Y}$ but the mutual information
$I(\SL A : \SL X, \SL Y) = 8 \eps^2 + \mathcal{O}(\eps^3)$ is positive.
Hence, this distribution violates Composition.
On the other hand, the parametrization technique from \cite{MarginalIndependence}
can also be used to describe the distributions on which $\CI{\SL A:\SL X,\SL Y}$
holds true. It is the codimension $1$ submodel parametrized by~\eqref{eq:Compo:Binary}
with $\eps = \delta \cdot \sfrac{\ol\alpha}{\alpha}$.
\end{example}

\begin{example}[Simple~matroids] \label{ex:Compo:Matroids}
Dually to \Cref{ex:Inter:Matroids}, we can consider the class of simple
matroids whose CI~structure satisfies Composition. Since there are no
loops or parallel elements in a simple matroid $(N, r)$, it satisfies
$\CI{i:j}$ for all distinct $i, j \in N$. Using Composition inductively
shows that $\CI{I:J}$ holds for all disjoint $I, J \subseteq N$.
It~follows that every element is a coloop and thus $(N, r)$ is the
free matroid.
\end{example}

\begin{example}[Linear spaces] \label{ex:Compo:Linear}
The argument from \Cref{ex:Compo:Matroids} relies on the rigid structure
of simple matroids, namely that all sets of $\text{size} \le 2$ must be
independent. For polymatroids this is no longer required as the rank of
a set $I$ is allowed to exceed its cardinality. Indeed, let $\mathbb{F}$
be any field and consider the vector space~$\SL V = \mathbb{F}^n$.
Take two distinct subspaces $\SL X, \SL Y$ which intersect non-trivially.
Then choose any subspace $\SL A$ such that $\SL A \cap (\SL X + \SL Y) =
\Set{0}$. The resulting subspace arrangement gives rise to an integer-valued
polymatroid (assigning to each collection of subspaces the dimension of
their span) satisfying $\CI{\SL A:\SL X, \SL Y}$ and $\CI[\mathrel{\centernot\CIperp}]{\SL X:\SL Y}$.
Hence, Composition is satisfied without the CI~structure being trivial.
\end{example}

\Cref{ex:Compo:Binary} shows that there are strictly positive distributions
which do not satisfy the Composition property. Thus, Composition does not
admit sufficient conditions which require a ``richness of support'' like
\Cref{cor:Inter:GK} in the case of Intersection.
Following Cartwright, Engström and Fink, one might hope that a sufficient
condition may still be hidden in the primary decomposition of the Composition~ideal,
even if it does not take the form of support constraints. Recall that the
Intersection ideal has one minimal prime for each admissible bipartite graph
on $Q_{\SL X} \sqcup Q_{\SL Y}$. This rich structure invites further investigation
which leads to \Cref{cor:Inter:GK}.
However, Kirkup~\cite{Kirkup} proved that the Composition ideal has only one
minimal prime whose variety contains any probability distribution at all.
Thus, there is no relevant structure in the primary decomposition and this
approach is also a dead end.

\subsection{The interaction information criterion}

A clearer line of attack is suggested by the information diagram
\Cref{fig:Assump}. Since the mutual information $I(\SL A:\SL X) = 0$,
it follows that $f = I(\SL A:\SL X\mid \SL Y) = {-I(\SL A:\SL X:\SL Y)}$.
In~particular, whenever the premises of Composition are satisfied,
the \emph{interaction information} (also known as \emph{triple mutual
information}) is non-positive. To obtain the conclusion, it must be
exactly zero. Hence, we have the following criterion.

\begin{lemma}[Interaction information criterion] \label{lemma:Compo:Interaction}
If $I(\SL A:\SL X:\SL Y) \ge 0$ then Composition holds: $\CI{\SL A:\SL X} \land \CI{\SL A:\SL
Y} \implies \CI{\SL A:\SL X, \SL Y}$.
Analogously, if $I(\SL A:\SL X:\SL Y) \le 0$ then we get Intersection:
$\CI{\SL A:\SL X|\SL Y} \land \CI{\SL A:\SL Y|\SL X} \implies \CI{\SL A:\SL X, \SL Y}$.
\end{lemma}

The sign of the interaction information $I(\SL A:\SL X:\SL Y)$ governs
monotonicity of the mutual information of any two variables under conditioning
on the third one:
\begin{align*}
  I(\SL A:\SL X:\SL Y)
  &= I(\SL A:\SL X) - I(\SL A:\SL X\mid \SL Y) \\
  &= I(\SL A:\SL Y) - I(\SL A:\SL Y\mid \SL X) \\
  &= I(\SL X:\SL Y) - I(\SL X:\SL Y\mid \SL A).
\end{align*}
Thus, its non-negativity yields the stronger implications
$\CI{i:j} \implies \CI{i:j|k}$.

Andrei Romashchenko kindly pointed out the following justification for why
this criterion is worth stating. Consider an interactive protocol where
Alice and Bob receive (possibly correlated) random seeds $\SL X$ and~$\SL Y$,
respectively. They exchange finitely many messages $\SL M_1, \SL M_2, \dots$
which are aggregated into a transcript.
Let $\SL A_k = (\SL M_1, \dots, \SL M_k)$ be the partial transcript at any
point in the communication. The next message $\SL M_{k+1}$ is sent by Alice,
say, after some local computations based only on her private data~$\SL X$
and the transcript~$\SL A_k$ containing all the information Bob has so far
revealed about his data~$\SL Y$. 
\begin{figure}
\begin{center}
\scalebox{0.7}{%
\begin{tikzpicture}[scale=0.5]
\node (a) at (0,0) {};
\node (b) at (4,0) {};
\node (c) at (2,3.464) {};
\draw (a) circle (3.7);
\draw (b) circle (3.7);
\draw (c) circle (3.7);

\node[shift={(0,0.9)}] (A|XY) at (c) {\Large$0$};
\node[shift={(-0.9,-0.4)}] (X|AY) at (a) {\Large$*$};
\node[shift={(0.9,-0.4)}] (Y|AX)  at (b)  {\Large$*$};
\node (A:X:Y) at (2,1.1) {\Large$h$};
\node (X:Y|A) at (2,-1.2) {\Large$*$};
\node (A:X|Y) at (-0.1,2.4) {\Large$*$};
\node (A:Y|X) at (4.1,2.4) {\Large$0$};

\node (A) at (2,8) {\Large$\SL M_{k+1}\mid \SL A_k$};
\node (X) at (-4.5,-3) {\Large$\SL X\mid \SL A_k$};
\node (Y) at (8.5,-3) {\Large$\SL Y\mid \SL A_k$};
\end{tikzpicture}}
\end{center}
\caption{Information diagram of Alice, Bob and the next message in an
interactive protocol, conditional the transcript of their communication.}
\label{fig:Interactive}
\end{figure}
This~implies~that ${H(\SL M_{k+1}\mid \SL X, \SL A_k)} = 0$
and it follows that $h = {I(\SL X:\SL Y:\SL M_{k+1} \mid \SL A_k)}
= {I(\SL M_{k+1}:\SL Y \mid \SL A_k)}$ is non-negative and hence
$I(\SL X:\SL Y \mid \SL A_k) \ge I(\SL X:\SL Y\mid \SL A_{k+1})$; see~\Cref{fig:Interactive}.
Inductively on the length of the transcript this leads to the well-known
conclusion that communication can only decrease mutual information.
In the communication complexity literature, this result is also stated
as: $\text{external information cost} \ge \text{internal information cost}$,
see~\cite{CompressCommunication}.
Hence, interactive two-party protocols give rise to joint distributions
in which the interaction information involving the transcript is
non-negative and hence Composition~holds.

The following support structure condition for jointly distributed
$\SL A, \SL X, \SL Y$ was studied by Kaced, Romashchenko and
Vereshchagin in \cite{IneqComb}:
\begin{gather}
  \label{eq:KRV}
  \begin{aligned}
  &\text{for all $(x,y) \in Q_{\SL X} \times Q_{\SL Y}$ there is at most one value $a \in Q_{\SL A}$} \\
  &\text{such that ${\Pr[\SL A = a, \SL X = x]} > 0$ and ${\Pr[\SL A = a, \SL Y = y]} > 0$.}
  \end{aligned}
\end{gather}
The assumption \eqref{eq:KRV} (the \emph{KRV condition}) ensures that $H(\SL A\mid \SL X)
+ H(\SL A\mid \SL Y) \le H(\SL A)$. From this it can easily be deduced
that $I(\SL A:\SL X:\SL Y) \ge H(\SL A\mid \SL X, \SL Y) \ge 0$.
In fact, the KRV condition also forces $H(\SL A\mid \SL X, \SL Y) = 0$
in the previous chain of inequalities.
Combining this with \Cref{lemma:Compo:Interaction} yields:

\begin{theorem}[{\cite[Theorem~1]{IneqComb}}]
If $\SL A, \SL X, \SL Y$ satisfy the KRV condition~\eqref{eq:KRV}
then ${H(\SL A\mid \SL X)} + {H(\SL A\mid \SL Y)} \le H(\SL A)$ and hence
$\CI{\SL A:\SL X} \land \CI{\SL A:\SL Y} \implies \CI{\SL A:\SL X, \SL Y}$.
\end{theorem}

\subsection{The dual conditional Ingleton criterion}

The Composition property can also be characterized synthetically using
CI~implications derived from conditional Ingleton inequalities as in
\Cref{sec:Inter:CondIng}.

\begin{theorem}[Dual conditional Ingleton criterion] \label{thm:Compo:CondIng}
Let $\SL A, \SL X, \SL Y, \SL G$ be jointly distributed discrete random
variables satisfying any of the following conditions:
\begin{enumerate}[label=(\roman*), itemsep=0.2em, leftmargin=4em]
\item\label{thm:Compo:CondIng:X}$\CI{\SL A:\SL G|\SL X,\SL Y}$ and $\CI{\SL X:\SL Y|\SL A}$,
\item\label{thm:Compo:CondIng:Y}$\CI{\SL X:\SL G|\SL A,\SL Y}$ and $\CI{\SL A:\SL Y|\SL X}$, or
\item\label{thm:Compo:CondIng:Z}$\CI{\SL Y:\SL G|\SL A,\SL X}$ and $\CI{\SL A:\SL X|\SL Y}$.
\end{enumerate}
Then $\CI{\SL A:\SL X|\SL G} \land \CI{\SL A:\SL Y|\SL G} \implies
\CI{\SL A:\SL X, \SL Y|\SL G}$ holds.
\end{theorem}

\begin{proof}
Analogously to the proof of \Cref{thm:Inter:CondIng}, the conditions~\ref{thm:Compo:CondIng:Y}
and~\ref{thm:Compo:CondIng:Z} are symmetric and they follow from (I:14)
in \cite{StudenyIngleton}. Condition~\ref{thm:Compo:CondIng:X} follows
from~(I:19).
\end{proof}

\begin{example}
For any pair $\SL X, \SL Y$ one can construct via \Cref{ex:Compo:MatrixFactor}
a common cause~$\SL A$ such that $\CI{\SL X:\SL Y|\SL A}$ holds. Then take any
function $\SL G$ of $(\SL X, \SL Y)$. The resulting joint distribution has
$\CI{\SL A:\SL G|\SL X, \SL Y}$ and thus \Cref{thm:Compo:CondIng}~\ref{thm:Compo:CondIng:X}
applies.
\end{example}

Note that in \Cref{thm:Compo:CondIng}, the Composition property is obtained
conditional on the auxiliary~$\SL G$. Some care must be taken in interpreting
this result: It is not guaranteed that every conditional distribution
$(\SL A, \SL X, \SL Y \mid \SL G = g)$ satisfies the Composition property!
Rather, the theorem promises that if \emph{all} conditional distributions
satisfy the premises of Composition (i.e., $\CI{\SL A:\SL X|{\SL G=g}}$ and
$\CI{\SL A:\SL Y|{\SL G=g}}$ for all $g \in Q_{\SL G}$), then they also
satisfy the conclusion. However, if one of the conditional distributions
does not satisfy the premises, the theorem allows another conditional
distribution to satisfy the premises without the conclusion.
This makes the criterion appear to be somewhat harder to work with as it
requires a suitable coupling of the conditional distributions through~$\SL G$.

\begin{remark}
Every compositional graphoid satisfying $\CI{\SL A:\SL X|\SL G}$,
$\CI{\SL A:\SL Y|\SL G}$ and \Cref{thm:Compo:CondIng}~\ref{thm:Compo:CondIng:X}
must even have $\CI{\SL A:\SL X, \SL Y, \SL G}$ and $\CI{\SL X:\SL Y}$.
\end{remark}

\section{Remarks}
\label{sec:Remarks}

\paragraph{Duality.}

Intersection and Composition are not only converses modulo the
semigraphoid axioms but also \emph{dual}. For an elementary
CI~statement $\CI{i:j|L}$ over ground set $N$, the \emph{dual
statement} is $\CI{i:j|L}^* \defas \CI{i:j|N \setminus ijL}$.
Applying duality statement-wise transforms
\begin{alignat*}{2}
  \text{\bfseries Intersection\hphantom{${}^*$}}\quad && \CI{i:j|kL}  \land \CI{i:k|jL}  &\implies \CI{i:j|L} \land \CI{i:k|L} \; \text{ into}\\
  \text{\bfseries Intersection${}^*$}\quad            && \CI{i:j|\PL} \land \CI{i:k|\PL} &\implies \CI{i:j|k\PL} \land \CI{i:k|j\PL},
\end{alignat*}
where $\PL = N \setminus ijkL$. But this is the Composition property with
$L$ replaced by~$\PL$. Hence, the dual of a CI~structure satisfying
Intersection is a CI~structure satisfying Composition and vice versa.
Remarkably, the sets of sufficient conditions obtained in
\Cref{thm:Inter:CondIng,thm:Compo:CondIng} are also formally dual to
each other. This is a feature of the conditional information inequalities
used in their proofs, although in general it is not true that any valid
conditional information inequality can be dualized and remain valid.

Denote by $\SR I_4$ the set of CI~structures which are representable
by four discrete random variables and satisfy all instances of Intersection;
analogously $\SR C_4$ for the Composition property. It can be verified that
$\lvert \SR I_4\vert = \lvert \SR C_4\rvert = 5\,736$ and they both have
the same number of elements modulo the action of the symmetric group $S_4$
on the random variables. These orbits are usually called \emph{permutational
types} and $\SR I_4$ and $\SR C_4$ both have $369$ of them. However, this
coincidence of numbers is \emph{not} explained by duality. For example,
the CI~structure of the distribution in \cite[Example~4]{StudenyIngleton},
is in $\SR I_4 \cap \SR C_4$ but its dual is not probabilistically
representable as it violates Studený's rule~(I:1).

Moreover, the sets $\SR I_4$ and $\SR C_4$ have a natural lattice structure.
Using code adapted from \cite{Selfadhe}, we have computed that $\SR I_4$
has $23$ permutational types of irreducible elements and that $\SR C_4$ has
$24$ such permutational types. Hence, the lattices are not isomorphic.
In view of this incompatibility, we believe that the coincidence of the
cardinalities of $\SR I_4$ and $\SR C_4$ is an artifact of the small ground
set size rather than a reflection of a deeper connection between the
two properties.

\paragraph{Operational characterizations of \Cref{thm:Inter:CondIng,thm:Compo:CondIng}.}

One of the merits of the Gács--Körner criterion for Intersection is that the
auxiliary variable~$\SL{GK}$ can be directly constructed and has an operational
interpretation as the common information of $\SL X$ and~$\SL Y$.
Both of these aspects have to be left unexplored in this article for the
auxiliary variables appearing in \Cref{thm:Inter:CondIng,thm:Compo:CondIng}.
It would be interesting to attach an operational meaning to these random
variables or to provide direct constructions, even in special cases.

\paragraph{The number of probabilistic compositional graphoids.}

The study of Intersection and Composition is fueled by applications in
graphical modeling. Amini, Aragam and Zhou \cite{Bryon} have recently
generalized Markov boundary techniques for the compact encoding of
CI~structures away from a concrete graphical representation relying
only on structural properties of compositional graphoids. This begs
the question, which was originally posed to me by Bryon Aragam in
private communication, whether there are significantly more statistical
models satisfying Intersection and Composition than there are
graphical~models.

Studený argues in \cite[Section~3.6]{Studeny} that any type of graph
with a fixed number of edge types (and no hyperedges) can produce at
most $2^{O(p(n))}$ distinct CI~structures on $n$ nodes, where $p$ is
some polynomial. He also proves that there are $2^{2^{\Omega(n)}}$
CI~structures which are \emph{probabilistic} (i.e., representable by
discrete random variables) and concludes that graphical models cannot
hope to capture all the nuances of probabilistic CI~structures.
We suggest a similar approach:

\begin{conjecture}
The number of compositional graphoids which are representable by $n$
discrete random variables is asymptotically $2^{2^{\Omega(n^\eps)}}$
for some $\eps > 0$ (or even $\eps > 1$).
\end{conjecture}

Studený's construction of many probabilistic CI~structures implicitly
relies on cycle matroids. By \Cref{ex:Inter:Matroids,ex:Compo:Matroids}
these examples will not in large enough numbers satisfy Intersection
or Composition.
Note that the number of compositional graphoids is known to grow with
$2^{\Omega(n 2^n)}$ but the construction in \cite{ConstructionMethods}
is purely combinatorial and has no associated random variables.

\paragraph{Relation to Gaussianity.}

Regular Gaussian distributions satisfy both, Intersection and Composition.
For the third implication $\CI{\SL A:\SL X|} \land \CI{\SL A:\SL X|\SL Y}
\implies \CI{\SL A:\SL X, \SL Y}$ briefly discussed in \Cref{sec:Intro},
note that the premises are symmetric under exchanging $\SL A$ and $\SL X$,
but the consequence is not. This means that $\CI{\SL A:\SL X}$ and
$\CI{\SL A:\SL X|\SL Y}$ may be derived from $\CI{\SL A:\SL X, \SL Y}$
as well as from $\CI{\SL X:\SL A, \SL Y}$ using the semigraphoid axioms.
A more symmetric formulation of the converse implication
\[
  \CI{\SL A:\SL X} \land \CI{\SL A:\SL X|\SL Y} \implies \CI{\SL A:\SL X, \SL Y} \lor \CI{\SL X:\SL A, \SL Y}
\]
is sometimes referred to as \emph{Weak transitivity} and is known to hold
for Gaussians~as~well.
This analysis suggests that the realm of Gaussian random variables is
usually more pleasant to work in as far as elementary properties of
conditional independence, such as the semigraphoid properties and their
converses, are concerned.

\paragraph{The third implication.}

The proper (unsymmetrized) form of the third converse implication
has been considered in the work of \cite{DawidOperations} which
features a sufficient condition derived from a generalization of
Basu's theorem. For discrete random variables, our approach of
using conditional information inequalities also applies.

\begin{theorem}
Let $\SL A, \SL X, \SL Y$ be jointly distributed discrete random
variables. If there exists a discrete $\SL G$ jointly distributed
with $\SL A, \SL X, \SL Y$ satisfying $\CI{\SL A:\SL Y|\SL G}$ and
$\CI{\SL Y:\SL G|\SL X}$, then $\CI{\SL A:\SL X} \land \CI{\SL A:\SL X|\SL Y}
\implies \CI{\SL A:\SL X, \SL Y}$ holds.
\end{theorem}

\begin{proof}
This follows from (I:7) in \cite{StudenyIngleton}.
\end{proof}

\paragraph{Conditional information inequalities for Composition.}

Matúš~\cite{MatusPiecewise} proved the following piecewise linear
conditional information inequality:
\begin{gather}
  \label{eq:MatusPiecewise}
  \begin{gathered}
  \Big[ H(\SL A\mid \SL X, \SL Y) = H(\SL X \mid \SL A, \SL Y) = H(\SL Y \mid \SL A, \SL X)
  = I(\SL A:\SL X) = I(\SL A:\SL Y) = 0 \Big] \\
  {} \;\implies\; \Big[ H(\SL X) = H(\SL Y) \ge \log\left\lceil\exp H(\SL A)\right\rceil \Big].
  \end{gathered}
\end{gather}
The conclusion in fact characterizes completely the entropy profiles satisfying
the premises. The proof crucially relies on the \emph{tightness assumptions}
that each variable is a function of the remaining two; cf.~\cite{MatusCsirmaz}.
Under these assumptions, Composition holds non-vacuously if and only if
$H(\SL A) = H(\SL X) = H(\SL Y) = 0$ which is not a useful condition.
One~may ask if there is a version of this inequality in which the tightness
constraints are lifted. More~precisely, if $H(\SL A\mid \SL X, \SL Y)$, say,
is a very small but positive number $\eps$, does the inequality hold at least
approximately with a perturbation of order $\eps$? Such a generalization might
help in deriving a new sufficient condition for Composition. Unfortunately,
it does not exist. We show that the tightness assumptions are \emph{essential}
in the sense of~\cite{CondInfo}.

\begin{figure}
\begin{center}
\scalebox{0.7}{%
\begin{tikzpicture}[scale=0.5]
\node (a) at (0,0) {};
\node (b) at (4,0) {};
\node (c) at (2,3.464) {};
\draw (a) circle (3.7);
\draw (b) circle (3.7);
\draw (c) circle (3.7);

\node[shift={(0,0.9)}] (A|XY) at (c) {\Large$\eps$};
\node[shift={(-0.9,-0.4)}] (X|AY) at (a) {\Large$0$};
\node[shift={(0.9,-0.4)}] (Y|AX)  at (b)  {\Large$0$};
\node (A:X:Y) at (2,1.1) {\Large$0$};
\node (X:Y|A) at (2,-1.2) {\Large$\alpha$};
\node (A:X|Y) at (-0.1,2.4) {\Large$0$};
\node (A:Y|X) at (4.1,2.4) {\Large$0$};

\node (A) at (2,8) {\Large$\SL A$};
\node (X) at (-4,-3) {\Large$\SL X$};
\node (Y) at (8,-3) {\Large$\SL Y$};
\end{tikzpicture}}
\end{center}
\caption{An almost-tight distribution with a large violation
of~\eqref{eq:MatusPiecewise}.}
\label{fig:Tight}
\end{figure}

Fix any two positive constants $\alpha, \eps$. There exists a random variable
$\SL X$ of entropy~$\alpha$; join to it an identical copy $\SL Y$ such that
$\Pr[\SL X = \SL Y] = 1$. Then take a random variable $\SL A$ of entropy~$\eps$
and join it to the pair $\SL X, \SL Y$ independently. The resulting joint
distribution has the information diagram depicted in \Cref{fig:Tight}.
As~$\eps \to 0$, the variable $\SL A$ becomes constant and the entropy
profile in the limit satisfies the conditions~of~\eqref{eq:MatusPiecewise}. \linebreak[4] 
For~each~fixed $\lambda > 0$ there exists a distribution in this family
which satisfies ${H(\SL X \mid \SL A, \SL Y)} = {H(\SL Y \mid \SL A, \SL X)}
= {I(\SL A:\SL X)} = {I(\SL A:\SL Y)} = 0$ and $H(\SL X) = H(\SL Y) =
\alpha$ but also
\[
  \alpha + \lambda H(\SL A\mid \SL X, \SL Y) = \alpha + \lambda \eps \ll
  \log(2) \le \log\left\lceil\exp \eps\right\rceil = \log\left\lceil\exp H(\SL A)\right\rceil.
\]
Therefore, tightness is essential in \eqref{eq:MatusPiecewise}.
It remains an interesting open problem to characterize the entropy
profiles on three random variables satisfying only $I(\SL A:\SL X)
= I(\SL A:\SL Y) = 0$. This characterization would encompass every
sufficient and necessary condition for Composition which can be
formulated in terms of information quantities.

\paragraph{Sufficient conditions on $\SL X$ and $\SL Y$ alone.}

A striking feature of the Gács--Körner condition in \Cref{thm:Inter:GK}
is that it pertains to the distribution of the pair~$(\SL X, \SL Y)$.
If~$G(\SL X, \SL Y)$ is connected then it is impossible to construct a
third random variable~$\SL A$ such that Intersection is violated.
This leads to an interesting tangential question:

\begin{question}
For which $\SL X$ and $\SL Y$ is it possible to construct $\SL A$
such that $I(\SL A:\SL X) = I(\SL A:\SL Y) = 0$ but $I(\SL A:\SL X, \SL Y) > 0$?
\end{question}

\begin{figure}
\begin{center}
\scalebox{0.7}{%
\begin{tikzpicture}[scale=0.5]
\node (a) at (0,0) {};
\node (b) at (4,0) {};
\node (c) at (2,3.464) {};
\draw (a) circle (3.7);
\draw (b) circle (3.7);
\draw (c) circle (3.7);

\node[shift={(0,0.9)}] (A|XY) at (c) {\Large$0$};
\node[shift={(-0.9,-0.4)}] (X|AY) at (a) {\Large$0$};
\node[shift={(0.9,-0.4)}] (Y|AX)  at (b)  {\Large$0$};
\node (A:X:Y) at (1.85,1.2) {\Large$-g$};
\node (X:Y|A) at (2,-1.2) {\Large$g$};
\node (A:X|Y) at (-0.1,2.4) {\Large$x$};
\node (A:Y|X) at (4.1,2.4) {\Large$y$};

\node (A) at (2,8) {\Large$\SL A$};
\node (X) at (-4,-3) {\Large$\SL X$};
\node (Y) at (8,-3) {\Large$\SL Y$};
\end{tikzpicture}}
\end{center}
\caption{When $\SL A = \SL X + \SL Y$ for $\SL X$ and $\SL Y$ independent
and uniformly distributed in an abelian group then $x = y = g$.}
\label{fig:Group}
\end{figure}

\begin{example}
Suppose that $\SL X$ and $\SL Y$ are independent, uniformly distributed
on their respective support, and have the same entropy. Thus their supports
must have the same size and we may assume that they both range in the
finite group $(\BB Z/k\BB Z, +)$. The variable $\SL A = \SL X + \SL Y$
is also uniformly distributed in $\BB Z/k\BB Z$ and any two of the three
variables determine the third. Together with the independence of $\SL X$
and $\SL Y$ we get the situation in \Cref{fig:Group}.
From $H(\SL X) = H(\SL Y) = H(\SL A)$ we get that $x = y = x+y-g$
which implies $x = y = g$ and thus $I(\SL A:\SL X) = I(\SL A:\SL Y) = 0$.
Since $\SL A$ is a function of $\SL X, \SL Y$, we get a violation of
Composition from $I(\SL A:\SL X, \SL Y) = H(\SL A) = \log k$.
\end{example}

While our assumptions in the above construction are very strong, they are
all of ``special position'' type: independence, uniformity and equality
of entropies.  If a sufficient condition for Composition only in terms of
the $(\SL X, \SL Y)$-marginal exists then it cannot apply to these
distributions. It must, at least implicitly, either enforce dependence,
biased marginals or make assertions about the alphabets of $\SL X$ and
$\SL Y$.

\section*{Acknowledgements}

\setlength{\intextsep}{5pt}%
\setlength{\columnsep}{5pt}%
\begin{wrapfigure}{R}{0.18\linewidth}
\vspace{-.5\baselineskip}%
\centering%
\href{https://doi.org/10.3030/101110545}{%
\includegraphics[width=0.9\linewidth]{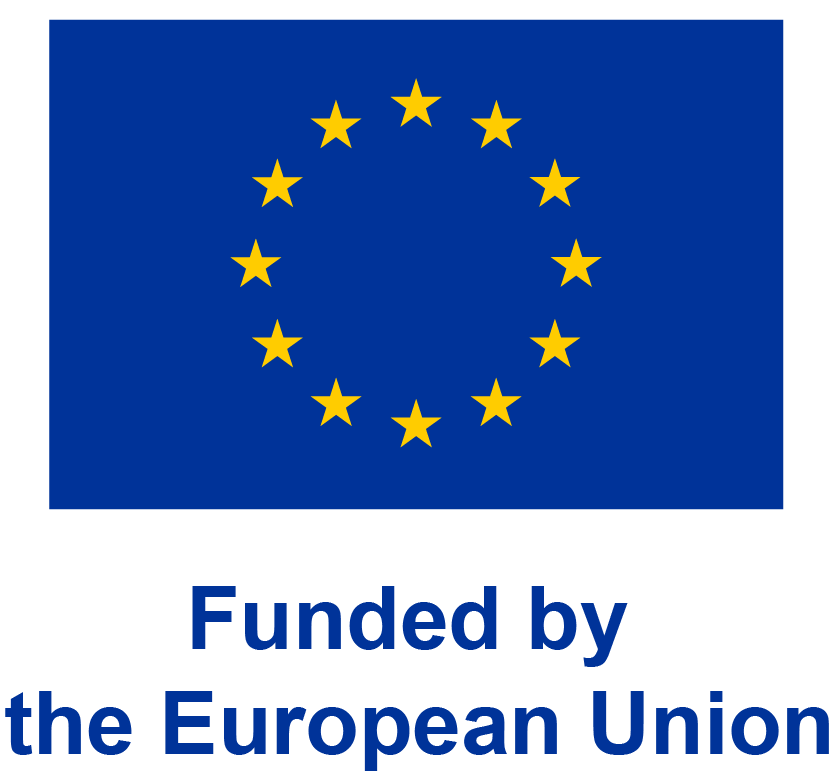}%
}
\end{wrapfigure}
I would like to thank Bryon Aragam and Andrei Romashchenko for enlightening
discussions around the Composition property. I also thank the anonymous
referees for their careful reading of the manuscript and suggestions for
improvement.
This research was funded by the European Union's Horizon 2020 research
and innovation programme under the Marie Skłodowska-Curie grant agreement
No.~101110545.


\begin{thebibliography}{10}

\bibitem{MaxLinearCI}
Carlos Am{\'e}ndola, Claudia Kl{\"u}ppelberg, Steffen Lauritzen, and Ngoc~M.
  Tran.
\newblock Conditional independence in max-linear {Bayesian} networks.
\newblock {\em Ann. Appl. Probab.}, 32(1):1--45, 2022.

\bibitem{Bryon}
Arash~A. Amini, Bryon Aragam, and Qing Zhou.
\newblock A non-graphical representation of conditional independence via the
  neighbourhood lattice, 2022.

\bibitem{Artin}
Emil Artin.
\newblock {\em Geometric algebra}, volume~3 of {\em Intersci. Tracts Pure Appl.
  Math.}
\newblock Interscience Publishers, 1957.

\bibitem{CompressCommunication}
Boaz Barak, Mark Braverman, Xi~Chen, and Anup Rao.
\newblock How to compress interactive communication.
\newblock {\em SIAM J. Comput.}, 42(3):1327--1363, 2013.

\bibitem{Selfadhe}
Tobias Boege, Janneke~H. Bolt, and Milan Studený.
\newblock Self-adhesivity in lattices of abstract conditional independence
  models.
\newblock {\em Discrete Applied Mathematics}, 361:196--225, 2025.

\bibitem{Maxoids}
Tobias Boege, Kamillo Ferry, Benjamin Hollering, and Francesco Nowell.
\newblock Polyhedral aspects of maxoids, 2025.

\bibitem{ConstructionMethods}
Tobias {Boege} and Thomas {Kahle}.
\newblock {Construction Methods for Gaussoids}.
\newblock {\em Kybernetika}, 56(6):1045--1062, 2020.

\bibitem{MarginalIndependence}
Tobias Boege, Sonja Petrovi\'{c}, and Bernd Sturmfels.
\newblock Marginal independence models.
\newblock In {\em Proceedings of the 2022 International Symposium on Symbolic
  and Algebraic Computation}, ISSAC~'22, pages 263--271. Association for
  Computing Machinery (ACM), 2022.

\bibitem{BokowskiSturmfels}
J{\"u}rgen Bokowski and Bernd Sturmfels.
\newblock {\em Computational synthetic geometry}, volume 1355 of {\em Lecture
  Notes in Mathematics}.
\newblock Springer, 1989.

\bibitem{Xiangying}
Xiangying Chen.
\newblock {\em {Conditional Erlangen Program}}.
\newblock PhD thesis, OvGU Magdeburg, 2024.

\bibitem{CoxLittleOShea}
David~A. Cox, John Little, and Donal O'Shea.
\newblock {\em Ideals, varieties, and algorithms}.
\newblock Undergraduate Texts in Mathematics. Springer, 4th edition, 2015.

\bibitem{Laszlo}
Laszlo Csirmaz.
\newblock A short proof of the {G\'acs--K\"orner} theorem, 2023.

\bibitem{CsiszarKoerner}
Imre Csisz{\'a}r and J{\'a}nos K{\"o}rner.
\newblock {\em Information theory. {Coding} theorems for discrete memoryless
  systems}.
\newblock Cambridge University Press, 2nd ed. edition, 2011.

\bibitem{DawidOperations}
A.~Philip {Dawid}.
\newblock {Conditional independence for statistical operations}.
\newblock {\em {Ann. Stat.}}, 8:598--617, 1980.

\bibitem{DawidTheory}
Alexander~P. Dawid.
\newblock Conditional independence in statistical theory.
\newblock {\em J. Roy. Statist. Soc. Ser. B}, 41(1):1--31, 1979.

\bibitem{AlgStat}
Mathias Drton, Bernd Sturmfels, and Seth Sullivant.
\newblock {\em Lectures on algebraic statistics}, volume~39 of {\em Oberwolfach
  Semin.}
\newblock Birkh{\"a}user, 2009.

\bibitem{MTP2}
Shaun Fallat, Steffen Lauritzen, Kayvan Sadeghi, Caroline Uhler, Nanny Wermuth,
  and Piotr Zwiernik.
\newblock Total positivity in {Markov} structures.
\newblock {\em Ann. Stat.}, 45(3):1152--1184, 2017.

\bibitem{Fink}
Alex {Fink}.
\newblock {The binomial ideal of the intersection axiom for conditional
  probabilities}.
\newblock {\em {J. Algebr. Comb.}}, 33(3):455--463, 2011.

\bibitem{GacsKoerner}
Peter {G\'acs} and J{\'a}nos {K\"orner}.
\newblock {Common information is far less than mutual information}.
\newblock {\em {Probl. Control Inf. Theory}}, 2:149--162, 1973.

\bibitem{M2}
Daniel~R. Grayson and Michael~E. Stillman.
\newblock Macaulay2, a software system for research in algebraic geometry.
\newblock Available at \url{http://www2.macaulay2.com}.
\newblock Version 1.22.

\bibitem{Ingleton}
Aubrey~W. Ingleton.
\newblock Representation of matroids.
\newblock In Dominic J.~A. Welsh, editor, {\em {Combinatorial Mathematics and
  its Applications. Proceedings of a Conference held at the Mathematical
  Institute, Oxford, from 7--10 July, 1969}}, pages 149--167, 1971.

\bibitem{CompoML}
Yibo Jiang, Bryon Aragam, and Victor Veitch.
\newblock Uncovering meanings of embeddings via partial orthogonality.
\newblock In {\em Proceedings of the 37th International Conference on Neural
  Information Processing Systems}, NeurIPS '23. Curran Associates Inc., 2023.

\bibitem{CondInfo}
Tarik {Kaced} and Andrei {Romashchenko}.
\newblock {Conditional information inequalities for entropic and almost
  entropic points}.
\newblock {\em {IEEE Trans. Inf. Theory}}, 59(11):7149--7167, 2013.

\bibitem{IneqComb}
Tarik Kaced, Andrei Romashchenko, and Nikolai Vereshchagin.
\newblock A conditional information inequality and its combinatorial
  applications.
\newblock {\em IEEE Trans. Inf. Theory}, 64(5):3610--3615, 2018.

\bibitem{GMAlgebra}
Thomas {Kahle}, Johannes {Rauh}, and Seth {Sullivant}.
\newblock {Algebraic aspects of conditional independence and graphical models}.
\newblock In Marloes {Maathuis}, Mathias {Drton}, Steffen {Lauritzen}, and
  Martin {Wainwright}, editors, {\em Handbook of graphical models}, {Chapman \&
  Hall/CRC Handbooks of Modern Statistical Methods}, pages 61--80. CRC Press,
  2019.

\bibitem{Kirkup}
George~A. {Kirkup}.
\newblock {Random variables with completely independent subcollections}.
\newblock {\em {J. Algebra}}, 309(2):427--454, 2007.

\bibitem{KaieFactor}
Robert Krone and Kaie Kubjas.
\newblock Uniqueness of nonnegative matrix factorizations by rigidity theory.
\newblock {\em SIAM J. Matrix Anal. Appl.}, 42(1):134--164, 2021.

\bibitem{KuehneYashfe}
Lukas K{\"u}hne and Geva Yashfe.
\newblock On entropic and almost multilinear representability of matroids.
\newblock 2025.

\bibitem{ExactCommon}
Gowtham~R. Kumar, Cheuk~T. Li, and Abbas {El Gamal}.
\newblock Exact common information.
\newblock In {\em 2014 IEEE International Symposium on Information Theory},
  pages 161--165, 2014.

\bibitem{Lauritzen}
Steffen Lauritzen.
\newblock {\em {Graphical models}}, volume~17 of {\em {Oxford Statistical
  Science Series}}.
\newblock Oxford University Press, 1996.

\bibitem{UnifyingMarkov}
Steffen Lauritzen and Kayvan Sadeghi.
\newblock Unifying {M}arkov properties for graphical models.
\newblock {\em Ann. Statist.}, 46(5):2251--2278, 2018.

\bibitem{Li}
Cheuk~Ting Li.
\newblock Undecidability of network coding, conditional information
  inequalities, and conditional independence implication.
\newblock {\em IEEE Trans. Inf. Theory}, 69(6):3493--3510, 2023.

\bibitem{LnenickaMatus}
Radim Ln{\v e}ni{\v c}ka and Franti{\v s}ek Mat{\'u}{\v s}.
\newblock On {Gaussian} conditional independence structures.
\newblock {\em Kybernetika}, 43(3):327--342, 2007.

\bibitem{MaierDB}
David Maier.
\newblock {\em The theory of relational databases}.
\newblock Computer {Software} {Engineering} {Series}. {Pitman} {Publishing}
  {Ltd.}, 1983.

\bibitem{MatusAscending}
Franti{\v{s}}ek Mat{\'u}{\v{s}}.
\newblock {Ascending and descending conditional independence relations}.
\newblock In {\em Transactions of the 11th Prague Conference on Information
  Theory, Statistical Decision Functions and Random Processes}, volume~B, pages
  189--200, 1992.

\bibitem{MatusPiecewise}
Franti{\v{s}}ek Mat{\'u}{\v{s}}.
\newblock Piecewise linear conditional information inequality.
\newblock {\em IEEE Trans. Inf. Theory}, 52(1):236--238, 2006.

\bibitem{MatusCsirmaz}
Franti{\v{s}}ek Mat{\'u}{\v{s}} and L{\'a}szlo Csirmaz.
\newblock Entropy region and convolution.
\newblock {\em IEEE Trans. Inf. Theory}, 62(11):6007--6018, 2016.

\bibitem{Oxley}
James Oxley.
\newblock {\em Matroid theory}, volume~21 of {\em Oxford Graduate Texts in
  Mathematics}.
\newblock Oxford University Press, 2nd edition, 2011.

\bibitem{Palacin}
Daniel Palac{\'{\i}}n.
\newblock An introduction to stability theory.
\newblock In {\em Lectures in model theory}, pages 1--27. European Mathematical
  Society (EMS), 2018.

\bibitem{PearlPaz}
Judea Pearl and Azaria Paz.
\newblock {GRAPHOIDS}: {A} graph-based logic for reasoning about relevance
  relations, or {W}hen would x tell you more about y if you already know z.
\newblock Technical Report CSD-850038, UCLA Computer Science Department, 1985.

\bibitem{MarkovBoundary}
Jose~M. Pe{\~n}a, Roland Nilsson, Johan Bj{\"o}rkegren, and Jesper Tegn{\'e}r.
\newblock Towards scalable and data efficient learning of {Markov} boundaries.
\newblock {\em Int. J. Approx. Reasoning}, 45(2):211--232, 2007.

\bibitem{Peters}
Jonas Peters.
\newblock On the intersection property of conditional independence and its
  application to causal discovery.
\newblock {\em J. Causal Inference}, 3(1):97--108, 2015.

\bibitem{Reichenbach}
Hans Reichenbach.
\newblock {\em The direction of time}.
\newblock {University of California Press}, 1971.
\newblock Edited by Maria Reichenbach.

\bibitem{Ignorable}
Ernesto {San {Mart{\'\i}n}}, Michel {Mouchart}, and Jean-Marie {Rolin}.
\newblock {Ignorable common information, null sets and Basu's first theorem}.
\newblock {\em {Sankhy\=a}}, 67(4):674--698, 2005.

\bibitem{Studeny}
Milan Studen{\'y}.
\newblock {\em {Probabilistic Conditional Independence Structures}}.
\newblock Information Science and Statistics. Springer, 2005.

\bibitem{StudenyIngleton}
Milan {Studen\'y}.
\newblock {Conditional independence structures over four discrete random
  variables revisited: conditional {I}ngleton inequalities}.
\newblock {\em {IEEE Trans. Inf. Theory}}, 67(11):7030--7049, 2021.

\bibitem{InfoCrypto}
Himanshu Tyagi and Shun Watanabe.
\newblock {\em Information-theoretic cryptography}.
\newblock Cambridge University Press, 2023.

\bibitem{Yashfe}
Geva Yashfe.
\newblock On the recognition problem for limits of entropy functions, 2025.

\bibitem{Yeung}
Raymond~W. {Yeung}.
\newblock {\em {A first course in information theory}}.
\newblock Information Technology: Transmission, Processing and Storage.
  Springer, 2005.

\end{thebibliography}

\makecontacts

\end{document}